%% file: template.tex
\documentclass[smallextended]{svjour3}       

\usepackage{tikz,lipsum,lmodern}
\usepackage[export]{adjustbox}
\usepackage{subcaption}
\usepackage{amsmath}
\usepackage{amssymb}
\usepackage{svg}
\usepackage[a4paper,margin=1.2in,footskip=0.25in]{geometry}
\usepackage[hidelinks]{hyperref}
\usepackage{xcolor}
\usepackage{graphicx}
\usepackage{pgfplots}
\pgfplotsset{compat=newest}
\usepackage{mathrsfs}
\usetikzlibrary{arrows,fit}
\newcommand{\orcid}[1]{\href{https://orcid.org/#1}{\includegraphics[width=10pt]{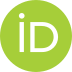}}}

\usepackage{cite}
\usepackage{setspace}
\usepackage{float}
\usepackage{booktabs}
\usepackage[english]{babel}
\begin{document}
\renewcommand{\arraystretch}{1.2}
\definecolor{zzttqq}{rgb}{0.6,0.2,0}
\definecolor{xdxdff}{rgb}{0.045019607843137253,0.45019607843137253,20}
\definecolor{uuuuuu}{rgb}{0.26666666666666666,0.26666666666666666,0.26666666666666666}
\title{Geometric genuine multipartite entanglement for four-qubit systems
}


\author{Ansh Mishra$^{1}$\orcid{0009-0003-4235-1714} \and Soumik Mahanti$^{*2,3}$\orcid{0000-0002-0380-0324} \and Abhinash Kumar Roy$^4$\orcid{0000-0001-7156-1989} \and Prasanta K. Panigrahi$^{2}$\orcid{0000-0001-5812-0353}}


\date{Received: date / Accepted: date}
\institute{
    Ansh Mishra \at
    \email{anshmishra471@gmail.com}
    \and
    Soumik Mahanti (Corresponding Author) \at
    \email{soumikmh1998@gmail.com}
    \and
    Abhinash Kumar Roy \at
    \email{aroy.iiser@gmail.com}
    \and
    Prasanta K. Panigrahi  \at
    \email{panigrahi.iiser@gmail.com}
    \and
$^1$Indian Institute of Science Education and Research, Mohali, Punjab 140306, India\\
$^2$Indian Institute of Science Education and Research Kolkata, Mohanpur, 741246 West Bengal, India\\
$^3$  S.N. Bose National Center for Basic Sciences, Block JD, Sector III, Salt Lake, Kolkata 700106, India\\
$^4$ Centre for Engineered Quantum Systems, Department of Physics and Astronomy, Macquarie University, Sydney, NSW 2122, Australia\\
$*$ the author to whom all the correspondence is to be made}

\maketitle

\begin{abstract}
Xie and Eberly introduced a genuine multipartite entanglement (GME) measure `concurrence fill'(\textit{Phys. Rev. Lett., \textbf{127}, 040403} (2021)) for three-party systems. It is defined as the area of a triangle whose side lengths represent squared concurrence in each bi-partition. However, it has been recently shown that concurrence fill is not monotonic under LOCC, hence not a faithful measure of entanglement. Though it is not a faithful entanglement measure, it encapsulates an elegant geometric interpretation of bipartite squared concurrences. There have been a few attempts to generalize GME measure to four-party settings and beyond. However, some of them are not faithful, and others simply lack an elegant geometric interpretation. The recent proposal from Xie et al. constructs a concurrence tetrahedron, whose volume gives the amount of GME for four-party systems; with generalization to more than four parties being the hypervolume of the simplex structure in that dimension. Here, we show by construction that to capture all aspects of multipartite entanglement, one does not need a more complex structure, and the four-party entanglement can be demonstrated using \textit{2D geometry only}. The subadditivity together with the Araki-Lieb inequality of linear entropy is used to construct a direct extension of the geometric GME to four-party systems resulting in quadrilateral geometry. Our measure can be geometrically interpreted as a combination of three quadrilaterals whose sides result from the concurrence in one-to-three bi-partition, and diagonal as concurrence in two-to-two bipartition. 
\end{abstract}

\section{Introduction}
Quantum correlations, in particular, entanglement is a fundamental feature of multipartite quantum systems, which arises from the superposition principle and tensor product structure of the Hilbert spaces \cite{schrodinger_1935, Horodecki_2009}. Entanglement is a crucial resource in various information processing tasks such as superdense coding \cite{Bennet_1992}, quantum teleportation \cite{Bennet_1993}, quantum key distribution \cite{Ekert91}, and quantum secret sharing \cite{Hillery_1999}, to mention a few. Therefore, the quantification of entanglement is of significant interest. For a two-party system, entanglement is quantified through various measures such as Concurrence \cite{Wootters_1997}, Negativity \cite{Jungnitsch_2011}, etc, and all faithful two-party entanglement measures agree in comparing entanglement for any two states \cite{singh2020revisiting}. However, classification and quantification of entangled states is still not well understood in higher dimensional and multiparty systems \cite{Bastin,Ribeiro,LiLi,Gilad,Ghahi,Home}. For more than two-party systems, owing to the presence of several possible bi-partitions, the entanglement structure becomes more complex; for instance, three-qubit states are of three types: fully separable, bi-separable, and genuinely entangled states \cite{Dur}. For information processing tasks involving multiple parties, genuinely entangled states act as a key resource \cite{Meyer_2002, Coffman_2000_distributed}, and hence there have been several studies to quantify multiparty entanglement \cite{Eberly_2021,Beckey_PRL_2021,shi2022genuine,guo2022genuine,PhysRevResearch.4.023059,brennen2003observable, Blasone_PRA,Sen_2010,Ma_2011, Coffman_2000}. Recently a seminal result has been obtained to show that there is no second law of entanglement manipulation, which implies that there is no unique function that governs all entanglement transformations \cite{Lami_23}. In short, the study of characterization and quantification of entanglement is a very rich area of research even today, particularly from a geometric perspective, which can provide insight into the distributive nature of entanglement \cite{Gharahi2,Eberly_2021,dharmaraj2022reimannian, mahanti2022classification}.

\noindent For three-party systems, there are various measures for quantifying genuine multipartite entanglement (GME), such as generalized geometric measure (GGM) \cite{Sen_2010}, genuine multipartite concurrence (GMC) \cite{Ma_2011}, and 3-tangle \cite{Coffman_2000} to name a few. However, each of them has its own limitations. 
The recent approach to quantify GME for three-qubit systems, termed as `concurrence fill' proposed by Xie and Eberly, has been argued to be superior to the other existing measures \cite{Eberly_2021}. It is shown to be zero for product and bi-separable states and positive for all genuinely entangled states. Furthermore, it leads to a fine-graining of the entanglement monotone and provides a tool for comparison of entanglement where other measures are inconclusive. The concurrence fill is a smooth function over the parametric space, and it takes into account entanglement across all bi-partitions to calculate the genuine multiparty entanglement. The proposed measure is also in accordance with the empirical results leading to an additional operational criterion for three qubits GME; namely, GHZ states being more genuinely entangled as compared to W states. In a recent article, Ge et al have constructed an explicit example of a local operation and classical communication (LOCC) scheme to prove that the area of the triangle increases under that operation, although the side lengths are nonincreasing \cite{Ge_PRA_2023}. This implies that `concurrence fill' is not a faithful measure of entanglement. However, they have found no counterexample of LOCC monotonicity if the sides of the triangle are taken as concurrence instead of squared concurrence. Thus they further conjecture that this `concurrence triangle' can be a faithful entanglement measure. Although concurrence fill is found not to be a faithful measure of entanglement, it truly captures many fascinating aspects of entanglement geometry and bipartite concurrences.

\noindent Recently, there have been several approaches to construct new GME measures incorporating all bipartite concurrences \cite{PhysRevResearch.4.023059,Yang_2022,guo2022genuine,shi2022genuine}. However, most of them fail to capture the elegant geometric interpretation, which was the main discovery in \cite{Eberly_2021}. In another recent work \cite{Eberly_2023}, the triangle-based measure has been generalized to the volume of a tetrahedron for four-party systems. They propose that for five-party pure states and more, the entanglement measure will be the hypervolume of the simplex in that dimension. Although the correctness of the approach is not questioned, we must point out that they have not proved the LOCC monotonicity of their proposed measure. Furthermore, for five-party and beyond, the simplex is hard to visualize and loses the elegant geometric representation. 

\noindent In this work, we present a similar measure for four-party systems inspired by geometry, but restrict it to 2 dimensions. We argue that to quantify genuine multiparty entanglement in four-party states, one need not consider higher dimensional geometric configurations. We show that the GME can be quantified through three planar structures, namely three quadrilaterals with each side representing \textit{one-to-three} squared concurrences. Now, four lengths do not fix the area of a quadrilateral as the number of free parameters for a quadrilateral is five. Hence, if the length of one diagonal of the quadrilateral is fixed, its area is fixed. Then the total area of the quadrilateral can be thought of as the sum of two triangles, with the diagonal as the common side. Here, we exploit the subadditivity of linear entropy to first show that the \textit{two-to-two} squared concurrence and combinations of two \textit{one-to-three} squared concurrences follow a triangle inequality. This relation allows us to treat the \textit{two-to-two} squared concurrence as the diagonals of the quadrilaterals. Furthermore, we show that if the concurrence across any bipartition is zero, then either the quadrilateral becomes of zero area, or it reduces to a triangle. Finally, to construct a GME, we multiply all the triangle areas. This is the simplest combination of smaller areas that satisfy all the GME criteria. We also point out that our measure can be constructed with both concurrence and squared concurrence as side lengths or diagonals. Following the conjecture in \cite{Ge_PRA_2023}, it can be the case that the GME formed using concurrence and squared concurrence show different performances under LOCC operations. It will be worth exploring the LOCC monotonicity for each case.

\noindent The manuscript is organized as follows. In Sec. \ref{section_2}, we revisit some preliminary concepts. In Sec. \ref{section_3}, we establish various constraints on the concurrences across different bi-partitions and show how it leads to the definition of our geometric measure of GME for all four party states. Moreover, we discuss its implications and practicality in comparison to other measures. In Sec. \ref{section_4}, we conclude by summarizing the results and future directions.
 
\section{Preliminaries}
\label{section_2}
Below we provide a review of a few concepts regarding entanglement and entanglement measures that are used throughout the paper.
\subsection*{\textit{I-Concurrence}}
Entanglement amounts to the correlation between different subsystems and acts as the most useful resource in most quantum information and communication protocols. Hence, the study of entanglement is of both foundational and operational importance \cite{Rather_PRL_2021}. An entanglement measure is a function $E(\rho)$ which maps the state space of the system to the positive real numbers, and thus a quantifier of entanglement. Any faithful measure $E(\rho)$ must satisfy certain conditions \cite{Vedral_1997}, which are following
\begin{enumerate}
    \item $E(\rho) = 0$ for separable states (separable states are not entangled). 
    \item $E(\rho)$ is invariant under local unitary transformations (local operations can not influence nonlocal correlation).
    \item $\langle E(\rho)\rangle \text{ }\geq \text{ }\langle E(\varepsilon_m(\rho))\rangle$ for any m-partite LOCC represented as $\varepsilon_m$ (average entanglement does not increase under LOCC).
\end{enumerate}
Concurrence is one such faithful measure of entanglement that is defined over any bipartition for pure as well as mixed states. The I-concurrence can be a  generalization of concurrence for pure multipartite systems \cite{Rungta_PRA_2001}, and has a profound geometrical interpretation \cite{Bhaskara_QINP_2017,roy2021geometric,banerjee2020quantifying,roy2022coherence}. Consider a multiparty system described by the pure state $\rho = |\psi\rangle\langle\psi|$. Then the I-concurrence corresponding to the bi-partition $\mathcal{A}|\mathcal{B}$ is given by:
\begin{equation}
    \mathcal{C}_{\mathcal{A}|\mathcal{B}} = \sqrt{2(1-Tr(\rho_{\mathcal{A}}^{2}))} = \sqrt{2(1-Tr(\rho_{\mathcal{B}}^{2}))}
\end{equation}

\noindent where, $\rho_{\mathcal{A}} = Tr_{B}(\rho)$ and $\rho_{\mathcal{B}} = Tr_{A}(\rho)$ are the reduced density matrices corresponding to the subsystem $\mathcal{A}$ and subsystem $\mathcal{B}$ respectively.

\subsection*{\textit{Genuine Multipartite Entanglement}}
The concept of ``genuine'' entanglement means that the state is inseparable across all possible bi-partitions i.e., each qubit is entangled with all other qubits. Mathematically, genuinely entangled pure states cannot be written as a tensor product of states across any bi-partition. For pure three-qubit systems, there are four classes of possible states, namely product states, bi-separable states, W states, and GHZ states. The product states are separable across all bi-partitions, bi-separable states are separable across one bi-partition and the W and GHZ states cannot be broken down into a product of states across any bi-partition. Hence they are genuinely entangled. In the case of four qubit systems, the entanglement structure becomes more complex. This complexity is owed to the fact that there are two kinds of possible bi-partitions. To be specific, there are four \textit{one-to-three} bi-partitions and three \textit{two-to-two} bi-partitions. The product states are the ones that are separable across all seven bi-partitions. The number of SLOCC classes for four qubit system is infinite as the inequivalent classes constitute a continuous family, but there are nine different subfamilies of the four qubit entangled states \cite{Verstraete_four_qubits}. For example, four qubit generalized GHZ and HS (Higuchi and Sudbery) states belong to the $G_{abcd}$ family,
\[G_{abcd} = \frac{a+d}{2}(|0000\rangle + |1111\rangle) + \frac{a-d}{2}(|0011\rangle + |1100\rangle) + \frac{b+c}{2}(|0101\rangle + |1010\rangle) + \frac{b-c}{2}(|0110\rangle + |1001\rangle).\]
The generalized $|W\rangle_{4}$ state is the member of $L_{ab_{3}}$ family,
\begin{equation*}
    \resizebox{.99\hsize}{!}{$L_{ab_{3}} = a(|0000\rangle + |1111\rangle) + \frac{a+b}{2}(|0101\rangle + |1010\rangle) +    \frac{a-b}{2}(|0110\rangle + |1001\rangle) + \frac{i}{\sqrt{2}}(|0001\rangle + |0010\rangle + |0111\rangle + |1011\rangle).$}
\end{equation*}
\subsection*{\textit{Genuine Multipartite Entanglement Measure}}
 An entanglement measure is called a faithful GME measure if it detects all and only the genuinely entangled multiparty states. Therefore, a GME measure should satisfy some additional criteria over being a faithful entanglement measure \cite{Ma_et_al}. The conditions are the following 
\begin{list}{\textbullet}{}
    \item $E(\rho) = 0$ for all bi-separable states. 
    \item $E(\rho) > 0$ for all states that are genuinely entangled.
\end{list}
The `concurrence fill' is one such GME measure, barring the LOCC monotonicity for three-qubit systems. Next, we present our direct extension of this measure in four-party states.

\section{Genuine multipartite entanglement measure for four qubit systems}
\label{section_3}
In this section, we present a genuine multipartite entanglement measure for four qubit systems. Four qubit systems are the simplest non-trivial case, where we have two different types of bi-partitions, namely \textit{one-to-three} and \textit{two-to-two} bi-partitions. For instance, four-qubit systems with parties represented by $A, B, C$ and $D$, have seven bi-partite concurrences. Out of them the four concurrences $\mathcal{C}_{A|BCD},$ $\mathcal{C}_{B|CDA},$ $\mathcal{C}_{C|DAB},$ $\mathcal{C}_{D|ABC}$ belong to the former type, whereas, the other three concurrences $\mathcal{C}_{AB|CD},$ $\mathcal{C}_{BC|DA},$ and $\mathcal{C}_{AC|BD}$ belong to the latter type. In the following, we present some of the constraints satisfied by these concurrences, which enable us to geometrically characterize the genuine entanglement in four-qubit systems.

\begin{theorem}[Polygon inequality]\label{thm1}
  The \emph{one-to-three} concurrences satisfy the following inequality:
\begin{equation}
    \mathcal{C}_{i|jkl}^{2} \leq \mathcal{C}_{j|kli}^{2} + \mathcal{C}_{k|lij}^{2} + \mathcal{C}_{l|ijk}^{2}.
    \label{eq:polygon_inequality_4qubit}
\end{equation} where $\{i, j, k, l\} = \{A, B, C, D\}\text{ and }i\neq j \neq k \neq l$.  
\end{theorem}

 \begin{proof}
  The normalized Schmidt weight ($\mathcal{Y}$) \cite{Grobe_1994} satisfies the following  inequality for N-qubit systems \cite{Qian_2018}:
\begin{equation}
    \mathcal{Y}_{i} \leq \sum_{j \neq i}{\mathcal{Y}_{j}}
    \label{2}
\end{equation}
where $\mathcal{Y}_{k}$ is the normalized Schmidt weight for the bi-partition between $k$-th qubit and rest of the qubits. It has been proved that any monotonically increasing concave function of $\mathcal{Y}$ also satisfies the above inequality \cite{Qian_2018}. Thus, it is enough to prove that concurrence squared is a monotonically increasing concave function of $\mathcal{Y}$ to show that the aforementioned theorem is true. Concurrence can conveniently be written in terms of normalized Schmidt weight as follows:
\begin{equation}
    \mathcal{C}_{i|jkl}^{2} = \mathcal{Y}_{i}(2-\mathcal{Y}_{i})
\end{equation}
Given a function $f(x) = x(2-x)$. Then $f'(x) = 2(1-x)$ and $f''(x) = -2$. Since $f'(x) \geq 0$ and $f''(x) \leq 0$ in the domain $[0,1]$. Therefore, $f(x)$ is a monotonically increasing concave function for $x \in [0, 1]$. 
  
\end{proof}

\begin{theorem} \label{thm2}
Two one-to-three concurrences and the corresponding two-to-two concurrence form a closed figure, namely, it satisfies the triangle inequalities:
\begin{equation}
    \begin{aligned}
        &C^{2}_{i|jkl} + C^{2}_{j|ikl} \geq C^{2}_{ij|kl}\\
        &C^{2}_{ij|kl} + C^{2}_{i|jkl} \geq C^{2}_{j|ikl}\\
        &C^{2}_{ij|kl} + C^{2}_{j|ikl} \geq C^{2}_{i|jkl},
    \end{aligned}
\end{equation}
where $\{i, j, k, l\} = \{A, B, C, D\}\text{ and }i\neq j \neq k \neq l$.  
\end{theorem}
\begin{proof}
    Since we are dealing with the pure four-party states, where the concurrence corresponding to a bipartition is related to the linear entropy of subsystems as $C_{A|B} = \sqrt{2S_{L}(\rho_{A})}$. The first inequality follows from the subadditivity of linear entropy, i.e., $S_{L}(\rho_{i})+S_{L}(\rho_{j})\geq S_{L}(\rho_{ij})$ \cite{2008_PRA_linear_entropy,2020_Elsevier_Entropy}. The second and third inequality follows from the Araki-Lieb inequality for linear entropy \cite{2020_Elsevier_Entropy}, namely, $|S_{L}(\rho_{i})-S_{L}(\rho_{j})|\leq S_{L}(\rho_{ij})$, which is equivalent to two inequalities, i.e., $S_{L}(\rho_{i})+S_{L}(\rho_{ij})\geq S_{L}(\rho_{j})$ and $S_{L}(\rho_{j})+S_{L}(\rho_{ij})\geq S_{L}(\rho_{i})$, leading to the second and third inequalities for concurrences in the above theorem. The positivity of concurrences along with the inequalities suggest that the lengths $C^{2}_{i|jkl}$, $C^{2}_{j|ikl}$ and $C^{2}_{ij|kl}$ form a closed figure, triangle in this case. 
\end{proof}

\noindent It is important to note that for pure states, the linear entropy of both the subsystems are equal, i.e., $S_{L}(\rho_{ij}) = S_{L}(\rho_{kl})$, leading to $C^{2}_{ij|kl} = C^{2}_{kl|ij}$. Therefore, the same set of inequalities as in Theorem 2, will be satisfied by the concurrences $C^{2}_{k|ijl}$, $C^{2}_{l|ijk}$ and $C^{2}_{ij|kl}$. Therefore, the two triangles formed using sides $C^{2}_{i|jkl}$, $C^{2}_{j|ikl}$, $C^{2}_{ij|kl}$ and $C^{2}_{k|ijl}$, $C^{2}_{l|ijk}$, $C^{2}_{ij|kl}$ have a side in common. This fact along with the Theorem 1 will lead to a realization that the two-to-two squared concurrences can be interpreted as the diagonals of the quadrilaterals formed by one-to-three squared concurrence.

\subsection*{\textbf{Geometric Interpretation}}

\noindent We emphasize on this point that both Theorem \ref{thm1} and Theorem \ref{thm2} are valid for concurrence as well as squared concurrence. From theorem 1, it is evident that they form a triangle in a plane. Though suggested in \cite{Eberly_2021}, the interpretation of squared concurrences as a higher dimensional figure, such as a face of a tetrahedron, is not straightforward because fixing the face areas of the tetrahedron \textit{does not fix} its volume. A similar kind of drawback is there in the case of a general quadrilateral with the sides as the one-to-three bipartite concurrence squared. It is not possible to fix the area of a quadrilateral by fixing its sides. But if the length of a diagonal is fixed, the area of the quadrilateral is also fixed. From Theorem \ref{thm2}, we showed that one can form two triangles with the concurrence (or squared concurrence) along \textit{two-to-two} bipartition as a common side. This implies that the four \textit{one-to-three} concurrence (or squared concurrence) can be interpreted as four sides of a quadrilateral with that \textit{two-to-two} concurrence (or squared concurrence) as its diagonal. We further use Theorem \ref{thm2} twice with the other two \textit{two-to-two} concurrences (or squared concurrences) to form another two quadrilaterals. These three quadrilaterals include all three different permutations one can use with four sides to make a quadrilateral. The whole construction represents an elegant geometry of the bipartite concurrences (or squared concurrences) for the four qubit case. The following figures depict the construction.

\begin{figure}[h!]
    \centering
    
    \subfloat{\input{images/new_quad1}}
    \subfloat{\input{images/new_quad2}}
    \subfloat{\input{images/new_quad3}}
    
    \caption{A schematic diagram showing the representation of seven bipartite squared concurrences in the form of three quadrilaterals.}
        
    \label{fig:quads}
\end{figure}
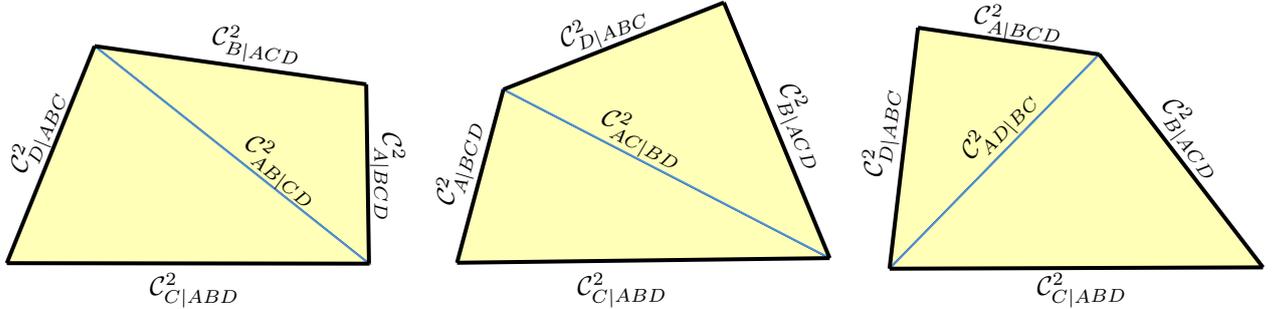

\noindent To form a GME measure out of the construction, we note the following observations. The relations from Theorem \ref{thm2}  
\begin{equation}\label{inequality}
\begin{split}
    &|C^2_{i|jkl}-C^2_{j|ikl}|\leq C^2_{ij|kl}\leq C^2_{i|jkl}+C^2_{j|ikl}\\
    &|C^2_{k|ijk}-C^2_{l|ijk}|\leq C^2_{ij|kl}\leq C^2_{k|ijl}+C^2_{l|ijk}
\end{split}
\end{equation}
 imply that if $C^2_{i|jkl} = 0$, then $C^2_{j|ikl}=C^2_{ij|kl}$. Geometrically this means if any side of the concurrence quadrilateral is zero, then its triangular adjacent side is equal to the diagonal. So, the quadrilaterals become triangles, shown in Fig \ref{Fig2}. And if $C^2_{ij|kl}=0$, then $C^2_{i|jkl}=C^2_{j|ikl}$ and $C^2_{l|ijk}=C^2_{k|ijk}$ along with $C^2_{ik|jl}=C^2_{il|jk}$. This implies that if any of the diagonals is zero, then that quadrilateral becomes collinear with two pairs of equal sides, and the other two diagonals are equal in magnitude (Fig. \ref{Fig3}). This makes the area of the quadrilateral to become zero, with the other two exactly equal. One simple GME construction out of this geometry is the following - the entanglement measure can be interpreted as the multiplication of areas of all the triangles. Thus, if any side of the triangle is zero, the whole thing will be zero. In other words, if concurrence across any of the bipartition is zero, the measure will give zero value. What is left to prove is that when all the sides and diagonal are nonzero, the measure must assign a nonzero value. We prove it in the following section. 
\subsection*{\textbf{Non-existence of certain configurations}}
Here, we show that certain geometric configurations of concurrences are not possible, which rules out the cases when all the bipartite concurrences are nonzero, but the geometric measure gives zero. We demonstrate it to be zero if and only if at least one of the concurrences is zero.
\begin{theorem} \label{thm3}
    If $C^{2}_{i|jkl}, C^{2}_{j|ikl}$ and $C^{2}_{ij|kl}$ are all non-zero, it is not possible to have the sum of one-to-three concurrences equal to the third, i.e, in Theorem \ref{thm2}, we will have strict inequality. 
\end{theorem}
\begin{proof}
    In order to show the above, we will use the following sharper inequality than subadditivity for linear entropy \cite{2020_Elsevier_Entropy}
    \begin{equation}\label{subadditivity}
        S_{L}(\rho_{12})\leq  S_{L}(\rho_{1}) + S_{L}(\rho_{2}) - 2(1-\sqrt{1-S_{L}(\rho_{1})})(1-\sqrt{1-S_{L}(\rho_{2})}).
    \end{equation}
    where $\rho_{12}$ is a bipartite state and $\rho_{1}$ and $\rho_{2}$ are reduced states. It is evident that the sum of entropies of the reduced system $S_{L}(\rho_{1}) + S_{L}(\rho_{2})$ upper bounds the right hand side of the above inequality. Therefore, for the inequality $C^{2}_{i|jkl} + C^{2}_{j|jkl} \geq C^{2}_{ij|kl}$ to saturate, we must have either $C^{2}_{i|jkl} = 0$ or $C^{2}_{i|jkl} = 0$. Therefore, in the case, when all three squared concurrences are non-zero, we have a strict inequality. Now we have to show that the LHS of the inequality in \ref{inequality} is also strict.     Using subadditivity of linear entropy,
    \[S_L(\rho_{j}) + S_L(\rho_{kl}) \geq S_L(\rho_{jkl}).\]
    Since we are dealing with pure four-party systems, we have $S_L(\rho_{kl})=S_L(\rho_{ij})$ and $S_L(\rho_{jkl})=S_L(\rho_{i})$. Thus, we can rewrite the above inequality as,
    \[S_L(\rho_{j}) + S_L(\rho_{ij}) \geq S_L(\rho_{i}).\]
   As shown earlier, this inequality is only saturated when either $S_L(\rho_{ij})$ or $S_L(\rho_{j})$ vanishes. Since all three of them are non-zero, therefore, this will be a strict inequality,
    \begin{equation}
        C^{2}_{i|jkl} - C^{2}_{j|ikl} < C^{2}_{ij|kl}
        \label{strict _inequality1}
    \end{equation}
    Following similar steps for the inequality,
    \[S_L(\rho_{i}) + S_L(\rho_{kl}) \geq S_L(\rho_{ikl}),\]
    and using the fact that we are working with a pure state four-party density matrix we get,
    \begin{equation}
        C^{2}_{j|ikl} - C^{2}_{i|jkl} < C^{2}_{ij|kl}
        \label{strict _inequality2}
    \end{equation}
    Combining equations \ref{strict _inequality1} and \ref{strict _inequality2} we get,
    \[|C^{2}_{i|jkl} - C^{2}_{j|ikl}| < C^{2}_{ij|kl}.\]
\end{proof}    
    The non-existence of configurations where the sum of two sides is equal to the third proves that the area of the individual triangles can not be zero when all sides are nonzero. This is a crucial step for our GME measure to be a proper measure.

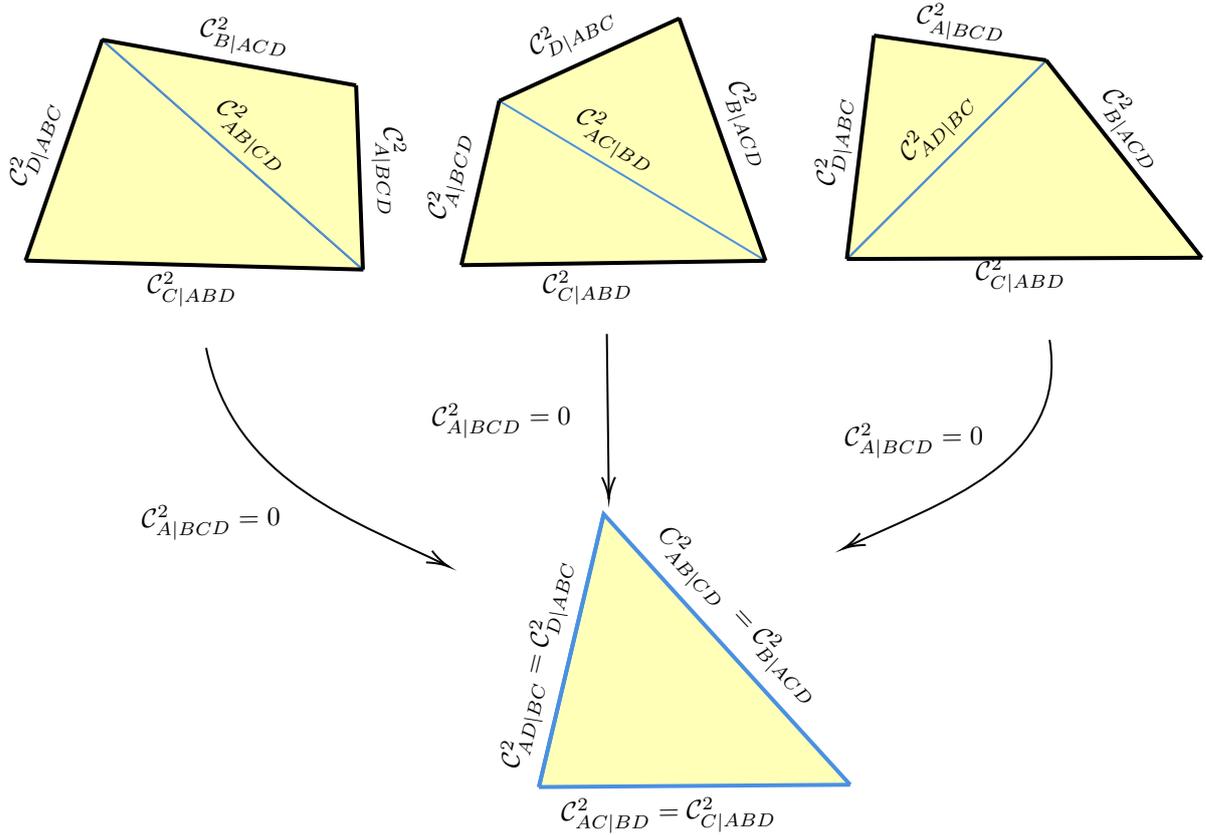
\begin{figure}[h!]
    \centering
    \input{images/triangle}
    \caption{If one party is separable from other parties, the four-party measure reduces to the `concurrence fill'}
    \label{Fig2}
\end{figure}
\noindent In addition to the above result, we also have the following theorem eliminating certain geometries of the quadrilateral formed using the one-to-three bipartitions.
\begin{theorem}
    If at least three of the squared concurrences from the \emph{one-to-three} bi-partition is nonzero, there cannot be a scenario where a particular squared concurrence is equal to the sum of the other three \textit{one-to-three} squared concurrence.
\end{theorem}
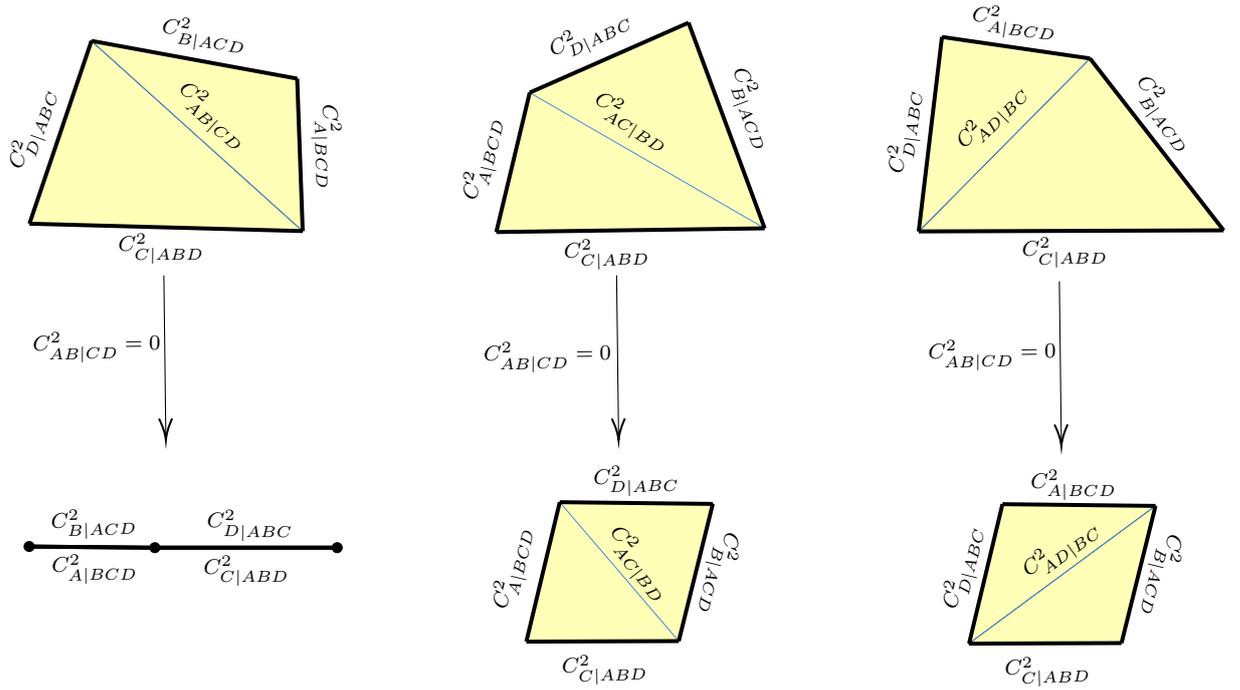
\begin{figure}
    \centering
    \input{images/diagonal_zero}
    \caption{If a four-party state is separable across a \textit{two-to-two} bipartition, the diagonal of a quadrilateral becomes zero. Thus one quadrilateral becomes a line.}
    \label{Fig3}
\end{figure}
\begin{proof}
 We shall do the proof by contradiction. As is well known, concurrence ($\mathcal{C}$) and normalized Schmidt weight ($\mathcal{Y}$) are related as \cite{Qian_2018},
\[\mathcal{C}^{2}(\mathcal{Y}) = \mathcal{Y}(2-\mathcal{Y}).\]
\noindent Now since $\mathcal{Y} \in [0, 1]$ we have,
\[\mathcal{Y}(\mathcal{C}^{2}) = 1 - \sqrt{1 - \mathcal{C}^{2}}\]
Now $\frac{\mathcal{Y}(\mathcal{C}^{2})}{\mathcal{C}^{2}}$ is a monotonically increasing function of $\mathcal{C}^{2}$.
Now for a four-qubit system, let us assume one of the $n$ \emph{one-to-rest} concurrences is equal to the sum of the rest. 
\begin{equation} \label{eq10}
\mathcal{C}_{i|jkl}^{2} = \mathcal{C}_{j|ikl}^{2}+\mathcal{C}_{k|ijl}^{2}+\mathcal{C}_{l|ijk}^{2}    
\end{equation}

\noindent Now if at least two terms in the RHS of Eq. \ref{eq10} are non-zero, we can write
\[\mathcal{C}_{a|rest}^{2}< \mathcal{C}_{i|jkl}^{2}, \text{where } a\in \{j,k,l\}\] 
We emphasize the strictly less than sign. Now, using the monotonicity of $\frac{\mathcal{Y}(\mathcal{C}^{2})}{\mathcal{C}^{2}}$,
\[\implies \frac{\mathcal{Y}(\mathcal{C}_{a|rest}^{2})}{\mathcal{C}_{a|rest}^{2}} < \frac{\mathcal{Y}(\mathcal{C}_{i|jkl}^{2})}{\mathcal{C}_{i|jkl}^{2}},\]
\[\implies \mathcal{Y}(\mathcal{C}_{a|rest}^{2}) < {\mathcal{C}_{a|rest}^{2}}\frac{\mathcal{Y}(\mathcal{C}_{i|jkl}^{2})}{\mathcal{C}_{i|jkl}^{2}}.\]
Summing over all $a$'s we obtain,
\begin{equation}
    \sum_{a=j,k,l}^{}{\mathcal{Y}(\mathcal{C}_{a|rest}^{2})} < \mathcal{Y}(\mathcal{C}_{i|jkl}^{2})
    \label{eqn:no-area-theorem1}
\end{equation}
However, from the polygon inequality \cite{Qian_2018},
\begin{equation}
     \sum_{a=j,k,l}^{}{\mathcal{Y}(\mathcal{C}_{a|rest}^{2})} \geq \mathcal{Y}(\mathcal{C}_{i|jkl}^{2})
    \label{eqn:no-area-thorem2}
\end{equation}
Eq. \ref{eqn:no-area-theorem1} is clearly contradicting Eq. \ref{eqn:no-area-thorem2}. Therefore, we cannot have a scenario where one squared concurrence in the \emph{one-to-three} bi-partition equals the sum of the rest of the squared concurrences if at least three of them are non-zero.   
\end{proof} 

\subsection*{\textit{The geometric measure of genuine four-partite entanglement}}
For a four-party system, we define the Genuine Multipartite Entanglement (GME) measure as the geometric mean of the areas of the six concurrence triangles constructed using both \textit{one-to-three} and \textit{two-to-two} bipartitions. Since the triangle inequality is followed both by the concurrence and squared concurrence, therefore we introduce two GME measures with side lengths as bipartite squared concurrence and concurrence respectively. The motivation follows from the findings in \cite{Ge_PRA_2023} that concurrence triangle is conjectured to be a faithful measure, but 'concurrence fill' is not. The quantification of genuine entanglement within the system is mathematically expressed as follows:

\begin{align}
& \mathcal{F}(|\psi\rangle_{ABCD}) \equiv \sqrt[6]{\prod_{i=1}^{6} \frac{4}{\sqrt{3}} \mathcal{A}_i}\\
& \mathcal{F}_1(|\psi\rangle_{ABCD}) \equiv \sqrt[6]{\prod_{i=1}^{6} \frac{4}{\sqrt{3}} \mathcal{B}_i}, 
\end{align}

\noindent where $\mathcal{A}_i$ denotes the area of the $i^{th}$ triangle formed by \textit{squared concurrences} as sides and $\mathcal{B}_i$ denotes the area of the $i^{th}$ triangle formed by the \textit{concurrences} as sides. Should the state under consideration be classified as either \emph{two-to-two} or \emph{one-to-three} separable, it is observed that a minimum of two out of the six triangle areas vanish. Consequently, this leads to a complete vanishing of the value for $\mathcal{F}(|\psi\rangle_{ABCD})$ and $\mathcal{F}_1(|\psi\rangle_{ABCD})$, rendering them equal to zero. Thus, the proposed measures yield zero values for all types of bi-separable states.

\noindent As widely acknowledged, concurrence remains invariant when subjected to local unitary transformations. Consequently, the geometric characteristics of triangles formed by the concurrence are unaffected by such transformations. As a result, the entanglement quantifiers $\mathcal{F}$ and $\mathcal{F}_1$ exhibit invariance under the influence of local unitary transformations and is symmetric under the permutation of parties. 

\noindent 
By scaling each of the triangle areas by the constant factor $\frac{4}{\sqrt{3}}$, a normalization condition is established for the four-qubit GHZ state, denoted as $|GHZ\rangle_{4}$, ensuring that $\mathcal{F}(|GHZ\rangle_{4}) = 1$ and $\mathcal{F}_1(|GHZ\rangle_{4}) = 1$. Consequently, for any states exhibiting less entanglement than the GHZ state, the resulting values of the entanglement measures $\mathcal{F}$ and $\mathcal{F}_1$ fall within the interval $0 \leq \mathcal{F}, \mathcal{F}_1 < 1$. Conversely, for states with higher degrees of entanglement than the GHZ state, the measures $\mathcal{F}$ and $\mathcal{F}_1$ exceed unity, i.e., $\mathcal{F}, \mathcal{F}_1 > 1$.

\noindent The following table summarizes the values of the entanglement measures $\mathcal{F}$ and $\mathcal{F}_1$ for specific quantum states, including the GHZ-state, W-state, four-qubit cluster state, and Higuchi-Sudbery state:

\begin{table}[h!]
\centering
\begin{tabular}{|l|l|l|}
\hline
                          & $\mathcal{F}$ & $\mathcal{F}_1$ \\ \hline
$|W\rangle_{4}$           & 0.646         & 0.817           \\ \hline
$|GHZ\rangle_{4}$         & 1.000         & 1.000           \\ \hline
$|\mathcal{C}\rangle_{4}$ & 1.095         & 1.077           \\ \hline
$|HS\rangle$              & 1.148         & 1.089           \\ \hline
\end{tabular}
\end{table}

\noindent The above quantum states are represented as follows: 
\[|W\rangle_{4} = \frac{1}{2}(|0001\rangle + |0010\rangle + |0100\rangle + |1000\rangle),\]
\[|GHZ\rangle_{4} = \frac{1}{\sqrt{2}}(|0000\rangle + |1111\rangle),\]
\[|\mathcal{C}\rangle_{4} = \frac{1}{2}(|0000\rangle + |0011\rangle + |1100\rangle - |1111\rangle)\text{, and}\]
\[|HS\rangle = \frac{1}{\sqrt{6}}[|0011\rangle + |1100\rangle + \omega(|0101\rangle + |1010\rangle) + \omega^{2}(|0110\rangle + |1001\rangle)].\]

\noindent The results indicate that the entanglement measures $\mathcal{F}$ and $\mathcal{F}_1$ correctly identify $|GHZ\rangle_{4}$ as more entangled than the $|W\rangle_{4}$ state. It is noteworthy that, in contrast to the three-qubit scenario, the four-qubit GHZ state is not the maximally entangled state. As a result, it is observed from the table that both the states $|HS\rangle$ and $|\mathcal{C}\rangle_{4}$ possess $\mathcal{F}, \mathcal{F}_1 > 1$, signifying that they are more entangled than the $|GHZ\rangle_4$ state, which aligns with the findings reported in reference \cite{HIGUCHI2000213}. For the case of cluster state $|\mathcal{C}\rangle_{4}$, concurrences across all the bipartitions are greater than or equal to the corresponding concurrences for $|GHZ\rangle_{4}$ state. Therefore, it is natural to expect that a genuine entanglement measure should yield a higher genuine entanglement for the cluster state as compared to the $|GHZ\rangle_{4}$ state. As summarized in the table, we observe a higher value for the cluster state than that of $|GHZ\rangle_{4}$ state for both the measures $\mathcal{F}$ and $\mathcal{F}_{1}$. This is in contrast with the measure based on the volume of tetrahedron, the symmetry in the entanglement value across all the bipartitions results in a higher value of GME for $|GHZ\rangle_{4}$ state than the cluster state. However, given that for four-party cases, we have two different kinds of bipartitions, it is not evident why symmetry across all the bipartitions should be a factor for higher entanglement. 

\noindent In the ref. \cite{Eberly_2021}, it was argued that for a three-party case, a GHZ state should be ranked higher in genuine multipartite entanglement as compared to the W state, on the grounds that the GHZ state is capable of certain quantum tasks better than the W state, in particular, deterministic teleportation of a single qubit state. It is worth noting here that a GHZ state in three qubit scenario is also an absolute maximally entangled state, i.e., it has maximal mixedness for the smaller subsystem in any considered bipartition. Therefore, the GHZ state having maximal genuine entanglement for three-party case is in line with it being an absolute maximally entangled state.

\noindent For the four-party case, the GHZ state fails to be an absolute maximally entangled state, (in fact, there is no absolute maximally entangled state in four qubit systems) because it is not maximally mixed in any two-to-two bipartitions. The cluster state, whereas, has the maximally mixed subsystems for all one-to-three bipartitions similar to the GHZ state. However, it is maximally mixed in two of the three two-to-two bipartitions. Therefore, the cluster state is closer to the absolute maximally entangled state as compared to the GHZ state, which is reflected in the higher values of both $\mathcal{F}$ and $\mathcal{F}_{1}$ for the cluster state. The measure proposed in \cite{Eberly_2023} fails to capture this aspect and attributes the highest value of GME to the GHZ state owing to its symmetry under the qubit permutations. Our proposed measure is, therefore, in line with the definition of an absolute maximally entangled state, which suggests that GHZ cannot be the maximally entangled state for four party case.
\section{Conclusion}
\label{section_4}
We proposed a genuine multipartite entanglement measure for four qubit systems, which takes into account entanglement across all possible bi-partitions. We have shown that we need not resort to higher dimensional geometric configurations (more than two) for a genuine multipartite entanglement measure for a four-party scenario. Through various triangle inequalities, we showed that the concurrences across bipartitions can be interpreted as forming closed planar structures, with their areas contributing to the genuine multipartite entanglement. We compared our measure with the other proposed measures and showed that in contrast to the earlier proposal, our measure is in line with the hierarchy suggested by absolute maximal entanglement, i.e., if in every bipartition, the smaller subsystem has more mixedness, the state should be more entangled, leading to the cluster state being more entangled as compared to GHZ in our measure. From the proposed measure, the constraints satisfied by the bi-partitions naturally arise out of the geometric structure. It is remarkable that even for four-party systems, the area of planar geometrical structures can give genuine entanglement in the system, and there is no need for any higher dimensional structure. It will be of significant interest to see whether genuine entanglement measures for arbitrary parties can give such simple geometrical representations. 

\begin{acknowledgements}
We would like to thank Mr. Aditya Raj and Mr. Abhishek Kumar for their earlier contribution to extending the four-party geometric GME. AM would like to thank IISER Kolkata for their hospitality. PKP acknowledges the support from DST, India, through Grant No. DST/ICPS/ QuST/Theme-1/2019/2020- 21/01.
\end{acknowledgements}

%
%

\bibliographystyle{spmpsci_unsort}      
\bibliography{sample}   

%
%

\end{document}

%% file: images/new_quad1.tex
\tikzset{every picture/.style={line width=0.75pt}} 

\begin{tikzpicture}[x=0.75pt,y=0.75pt,yscale=-0.82,xscale=0.82]

\draw [fill={rgb, 255:red, 255; green, 255; blue, 170 }  ,fill opacity=0.83 ]   (194.48,67.35) -- (361.55,200.8) -- (359.86,90.81) ;
\draw [color={rgb, 255:red, 74; green, 144; blue, 226 }  ,draw opacity=1 ][fill={rgb, 255:red, 255; green, 255; blue, 170 }  ,fill opacity=0.83 ]   (194.48,67.35) -- (361.55,200.8) -- (141.09,200.5) ;
\draw [color={rgb, 255:red, 0; green, 0; blue, 0 }  ,draw opacity=1 ][fill={rgb, 255:red, 255; green, 255; blue, 170 }  ,fill opacity=0.83 ][line width=1.5]    (194.48,67.35) -- (141.09,200.5) ;
\draw [color={rgb, 255:red, 0; green, 0; blue, 0 }  ,draw opacity=1 ][fill={rgb, 255:red, 255; green, 255; blue, 170 }  ,fill opacity=0.83 ][line width=1.5]    (361.55,200.8) -- (141.09,200.5) ;
\draw [color={rgb, 255:red, 0; green, 0; blue, 0 }  ,draw opacity=1 ][fill={rgb, 255:red, 255; green, 255; blue, 170 }  ,fill opacity=0.83 ][line width=1.5]    (359.86,90.81) -- (194.48,67.35) ;
\draw [color={rgb, 255:red, 0; green, 0; blue, 0 }  ,draw opacity=1 ][fill={rgb, 255:red, 255; green, 255; blue, 170 }  ,fill opacity=0.83 ][line width=1.5]    (359.86,90.81) -- (361.55,200.8) ;

\draw (385.22,118.74) node [anchor=north west][inner sep=0.75pt]  [color={rgb, 255:red, 0; green, 0; blue, 0 }  ,opacity=1 ,rotate=-89.86,xslant=0.02]  {$\mathcal{C}_{A|BCD}^{2}$};
\draw (266.34,50.44) node [anchor=north west][inner sep=0.75pt]  [color={rgb, 255:red, 0; green, 0; blue, 0 }  ,opacity=1 ,rotate=-8.74,xslant=0]  {$\mathcal{C}_{B|ACD}^{2}$};
\draw (226.57,204.67) node [anchor=north west][inner sep=0.75pt]  [color={rgb, 255:red, 0; green, 0; blue, 0 }  ,opacity=1 ,rotate=-0.18,xslant=0.01]  {$\mathcal{C}_{C|ABD}^{2}$};
\draw (137.54,141.88) node [anchor=north west][inner sep=0.75pt]  [color={rgb, 255:red, 0; green, 0; blue, 0 }  ,opacity=1 ,rotate=-293.33,xslant=-0.01]  {$\mathcal{C}_{D|ABC}^{2}$};
\draw (292.98,115.29) node [anchor=north west][inner sep=0.75pt]  [rotate=-40.1,xslant=-0.07]  {$\mathcal{C}_{AB|CD}^{2}$};

\end{tikzpicture}

%% file: images/new_quad2.tex
\tikzset{every picture/.style={line width=0.75pt}} 

\begin{tikzpicture}[x=0.75pt,y=0.75pt,yscale=-0.82,xscale=0.82]

\draw [fill={rgb, 255:red, 255; green, 255; blue, 170 }  ,fill opacity=0.83 ]   (468.96,186.26) -- (270.37,82.8) -- (242.13,189.11) ;
\draw [color={rgb, 255:red, 74; green, 144; blue, 226 }  ,draw opacity=1 ][fill={rgb, 255:red, 255; green, 255; blue, 170 }  ,fill opacity=0.83 ]   (468.96,186.26) -- (270.37,82.8) -- (404.57,29.47) ;
\draw [color={rgb, 255:red, 0; green, 0; blue, 0 }  ,draw opacity=1 ][fill={rgb, 255:red, 255; green, 255; blue, 170 }  ,fill opacity=0.83 ][line width=1.5]    (270.37,82.8) -- (404.57,29.47) ;
\draw [color={rgb, 255:red, 0; green, 0; blue, 0 }  ,draw opacity=1 ][fill={rgb, 255:red, 255; green, 255; blue, 170 }  ,fill opacity=0.83 ][line width=1.5]    (468.96,186.26) -- (242.13,189.11) ;
\draw [color={rgb, 255:red, 0; green, 0; blue, 0 }  ,draw opacity=1 ][fill={rgb, 255:red, 255; green, 255; blue, 170 }  ,fill opacity=0.83 ][line width=1.5]    (404.57,29.47) -- (468.96,186.26) ;
\draw [color={rgb, 255:red, 0; green, 0; blue, 0 }  ,draw opacity=1 ][fill={rgb, 255:red, 255; green, 255; blue, 170 }  ,fill opacity=0.83 ][line width=1.5]    (242.13,189.11) -- (270.37,82.8) ;

\draw (224.26,151.99) node [anchor=north west][inner sep=0.75pt]  [color={rgb, 255:red, 0; green, 0; blue, 0 }  ,opacity=1 ,rotate=-285.62,xslant=0.02]  {$\mathcal{C}_{A|BCD}^{2}$};
\draw (451.16,74) node [anchor=north west][inner sep=0.75pt]  [color={rgb, 255:red, 0; green, 0; blue, 0 }  ,opacity=1 ,rotate=-66.42,xslant=0]  {$\mathcal{C}_{B|ACD}^{2}$};
\draw (313.46,193.73) node [anchor=north west][inner sep=0.75pt]  [color={rgb, 255:red, 0; green, 0; blue, 0 }  ,opacity=1 ,rotate=-358.48,xslant=0.01]  {$\mathcal{C}_{C|ABD}^{2}$};
\draw (300.04,43.73) node [anchor=north west][inner sep=0.75pt]  [color={rgb, 255:red, 0; green, 0; blue, 0 }  ,opacity=1 ,rotate=-339.16,xslant=-0.01]  {$\mathcal{C}_{D|ABC}^{2}$};
\draw (335.1,89.05) node [anchor=north west][inner sep=0.75pt]  [rotate=-27.55,xslant=0.04]  {$\mathcal{C}_{AC|BD}^{2}$};

\end{tikzpicture}

%% file: images/new_quad3.tex
\tikzset{every picture/.style={line width=0.75pt}} 

\begin{tikzpicture}[x=0.75pt,y=0.75pt,yscale=-0.82,xscale=0.82]

\draw [fill={rgb, 255:red, 255; green, 255; blue, 170 }  ,fill opacity=0.83 ]   (258.66,170.32) -- (386.33,38.86) -- (276,22.5) ;
\draw [color={rgb, 255:red, 74; green, 144; blue, 226 }  ,draw opacity=1 ][fill={rgb, 255:red, 255; green, 255; blue, 170 }  ,fill opacity=0.83 ]   (258.66,170.32) -- (386.33,38.86) -- (486,169.5) ;
\draw [color={rgb, 255:red, 0; green, 0; blue, 0 }  ,draw opacity=1 ][fill={rgb, 255:red, 255; green, 255; blue, 170 }  ,fill opacity=0.83 ][line width=1.5]    (276,22.5) -- (258.66,170.32) ;
\draw [color={rgb, 255:red, 0; green, 0; blue, 0 }  ,draw opacity=1 ][fill={rgb, 255:red, 255; green, 255; blue, 170 }  ,fill opacity=0.83 ][line width=1.5]    (258.66,170.32) -- (486,169.5) ;
\draw [color={rgb, 255:red, 0; green, 0; blue, 0 }  ,draw opacity=1 ][fill={rgb, 255:red, 255; green, 255; blue, 170 }  ,fill opacity=0.83 ][line width=1.5]    (386.33,38.86) -- (486,169.5) ;
\draw [color={rgb, 255:red, 0; green, 0; blue, 0 }  ,draw opacity=1 ][fill={rgb, 255:red, 255; green, 255; blue, 170 }  ,fill opacity=0.83 ][line width=1.5]    (276,22.5) -- (386.33,38.86) ;

\draw (310.54,4.11) node [anchor=north west][inner sep=0.75pt]  [color={rgb, 255:red, 0; green, 0; blue, 0 }  ,opacity=1 ,rotate=-8.25,xslant=0.02]  {$\mathcal{C}_{A|BCD}^{2}$};
\draw (435.34,59.91) node [anchor=north west][inner sep=0.75pt]  [color={rgb, 255:red, 0; green, 0; blue, 0 }  ,opacity=1 ,rotate=-54.52,xslant=0]  {$\mathcal{C}_{B|ACD}^{2}$};
\draw (346.64,171.28) node [anchor=north west][inner sep=0.75pt]  [color={rgb, 255:red, 0; green, 0; blue, 0 }  ,opacity=1 ,rotate=-358.96,xslant=0.01]  {$\mathcal{C}_{C|ABD}^{2}$};
\draw (238.33,114) node [anchor=north west][inner sep=0.75pt]  [color={rgb, 255:red, 0; green, 0; blue, 0 }  ,opacity=1 ,rotate=-277.92,xslant=-0.01]  {$\mathcal{C}_{D|ABC}^{2}$};
\draw (295.9,95.23) node [anchor=north west][inner sep=0.75pt]  [rotate=-314.77,xslant=-0.07]  {$\mathcal{C}_{AD|BC}^{2}$};

\end{tikzpicture}

%% file: images/triangle.tex
\tikzset{every picture/.style={line width=0.75pt}} 

\begin{tikzpicture}[x=0.75pt,y=0.75pt,yscale=-1,xscale=1]

\draw [fill={rgb, 255:red, 255; green, 255; blue, 170 }  ,fill opacity=0.83 ]   (613,145.67) -- (535.32,46.34) -- (435.81,146.3) ;
\draw [fill={rgb, 255:red, 255; green, 255; blue, 170 }  ,fill opacity=0.83 ]   (395.25,147.19) -- (352.14,25.31) -- (262.29,66.76) ;
\draw [fill={rgb, 255:red, 255; green, 255; blue, 170 }  ,fill opacity=0.83 ]   (25.86,147.11) -- (64.08,36.04) -- (194.25,151.61) ;
\draw [color={rgb, 255:red, 74; green, 144; blue, 226 }  ,draw opacity=1 ][fill={rgb, 255:red, 255; green, 255; blue, 170 }  ,fill opacity=0.83 ][line width=1.5]    (437.47,410.74) -- (314.46,274.63) -- (282.08,411.91) ;
\draw [color={rgb, 255:red, 74; green, 144; blue, 226 }  ,draw opacity=1 ][fill={rgb, 255:red, 239; green, 239; blue, 36 }  ,fill opacity=0.83 ][line width=1.5]    (282.08,411.91) -- (437.47,410.74) -- cycle ;
\draw [color={rgb, 255:red, 74; green, 144; blue, 226 }  ,draw opacity=1 ][fill={rgb, 255:red, 239; green, 239; blue, 36 }  ,fill opacity=0.83 ][line width=1.5]    (314.46,274.63) -- (282.08,411.91) ;
\draw [color={rgb, 255:red, 74; green, 144; blue, 226 }  ,draw opacity=1 ][fill={rgb, 255:red, 255; green, 255; blue, 170 }  ,fill opacity=0.83 ]   (395.25,147.19) -- (262.29,66.76) -- (243.37,149.4) ;
\draw [color={rgb, 255:red, 0; green, 0; blue, 0 }  ,draw opacity=1 ][fill={rgb, 255:red, 255; green, 255; blue, 170 }  ,fill opacity=0.83 ][line width=1.5]    (262.29,66.76) -- (352.14,25.31) ;
\draw [color={rgb, 255:red, 0; green, 0; blue, 0 }  ,draw opacity=1 ][fill={rgb, 255:red, 255; green, 255; blue, 170 }  ,fill opacity=0.83 ][line width=1.5]    (395.25,147.19) -- (243.37,149.4) ;
\draw [color={rgb, 255:red, 0; green, 0; blue, 0 }  ,draw opacity=1 ][fill={rgb, 255:red, 255; green, 255; blue, 170 }  ,fill opacity=0.83 ][line width=1.5]    (352.14,25.31) -- (395.25,147.19) ;
\draw [color={rgb, 255:red, 0; green, 0; blue, 0 }  ,draw opacity=1 ][fill={rgb, 255:red, 255; green, 255; blue, 170 }  ,fill opacity=0.83 ][line width=1.5]    (243.37,149.4) -- (262.29,66.76) ;
\draw [color={rgb, 255:red, 74; green, 144; blue, 226 }  ,draw opacity=1 ][fill={rgb, 255:red, 255; green, 255; blue, 170 }  ,fill opacity=0.83 ]   (64.08,36.04) -- (194.25,151.61) -- (190.84,58.98) ;
\draw [color={rgb, 255:red, 0; green, 0; blue, 0 }  ,draw opacity=1 ][fill={rgb, 255:red, 255; green, 255; blue, 170 }  ,fill opacity=0.83 ][line width=1.5]    (64.08,36.04) -- (25.86,147.11) ;
\draw [color={rgb, 255:red, 0; green, 0; blue, 0 }  ,draw opacity=1 ][fill={rgb, 255:red, 255; green, 255; blue, 170 }  ,fill opacity=0.83 ][line width=1.5]    (194.25,151.61) -- (25.86,147.11) ;
\draw [color={rgb, 255:red, 0; green, 0; blue, 0 }  ,draw opacity=1 ][fill={rgb, 255:red, 255; green, 255; blue, 170 }  ,fill opacity=0.83 ][line width=1.5]    (190.84,58.98) -- (64.08,36.04) ;
\draw [color={rgb, 255:red, 0; green, 0; blue, 0 }  ,draw opacity=1 ][fill={rgb, 255:red, 255; green, 255; blue, 170 }  ,fill opacity=0.83 ][line width=1.5]    (190.84,58.98) -- (194.25,151.61) ;
\draw [color={rgb, 255:red, 74; green, 144; blue, 226 }  ,draw opacity=1 ][fill={rgb, 255:red, 255; green, 255; blue, 170 }  ,fill opacity=0.83 ]   (435.81,146.3) -- (535.32,46.34) -- (449.32,33.89) ;
\draw [color={rgb, 255:red, 0; green, 0; blue, 0 }  ,draw opacity=1 ][fill={rgb, 255:red, 255; green, 255; blue, 170 }  ,fill opacity=0.83 ][line width=1.5]    (449.32,33.89) -- (435.81,146.3) ;
\draw [color={rgb, 255:red, 0; green, 0; blue, 0 }  ,draw opacity=1 ][fill={rgb, 255:red, 255; green, 255; blue, 170 }  ,fill opacity=0.83 ][line width=1.5]    (435.81,146.3) -- (613,145.67) ;
\draw [color={rgb, 255:red, 0; green, 0; blue, 0 }  ,draw opacity=1 ][fill={rgb, 255:red, 255; green, 255; blue, 170 }  ,fill opacity=0.83 ][line width=1.5]    (535.32,46.34) -- (613,145.67) ;
\draw [color={rgb, 255:red, 0; green, 0; blue, 0 }  ,draw opacity=1 ][fill={rgb, 255:red, 255; green, 255; blue, 170 }  ,fill opacity=0.83 ][line width=1.5]    (449.32,33.89) -- (535.32,46.34) ;
\draw    (116,191) .. controls (128.87,259.31) and (194.67,279.6) .. (235.76,300.37) ;
\draw [shift={(237,301)}, rotate = 207.12] [color={rgb, 255:red, 0; green, 0; blue, 0 }  ][line width=0.75]    (10.93,-3.29) .. controls (6.95,-1.4) and (3.31,-0.3) .. (0,0) .. controls (3.31,0.3) and (6.95,1.4) .. (10.93,3.29)   ;
\draw    (537,187) .. controls (546.9,249.37) and (482.31,268.62) .. (435.42,291.31) ;
\draw [shift={(434,292)}, rotate = 333.92] [color={rgb, 255:red, 0; green, 0; blue, 0 }  ][line width=0.75]    (10.93,-3.29) .. controls (6.95,-1.4) and (3.31,-0.3) .. (0,0) .. controls (3.31,0.3) and (6.95,1.4) .. (10.93,3.29)   ;
\draw    (316,184) -- (316.98,265) ;
\draw [shift={(317,267)}, rotate = 269.31] [color={rgb, 255:red, 0; green, 0; blue, 0 }  ][line width=0.75]    (10.93,-3.29) .. controls (6.95,-1.4) and (3.31,-0.3) .. (0,0) .. controls (3.31,0.3) and (6.95,1.4) .. (10.93,3.29)   ;

\draw (216.99,78.4) node [anchor=north west][inner sep=0.75pt]  [color={rgb, 255:red, 0; green, 0; blue, 0 }  ,opacity=1 ,rotate=-88.22,xslant=0.02]  {$\mathcal{C}_{A|BCD}^{2}$};
\draw (113.22,21.38) node [anchor=north west][inner sep=0.75pt]  [color={rgb, 255:red, 0; green, 0; blue, 0 }  ,opacity=1 ,rotate=-8.74,xslant=0]  {$\mathcal{C}_{B|ACD}^{2}$};
\draw (85.32,151.09) node [anchor=north west][inner sep=0.75pt]  [color={rgb, 255:red, 0; green, 0; blue, 0 }  ,opacity=1 ,rotate=-0.18,xslant=0.01]  {$\mathcal{C}_{C|ABD}^{2}$};
\draw (12.91,107.17) node [anchor=north west][inner sep=0.75pt]  [color={rgb, 255:red, 0; green, 0; blue, 0 }  ,opacity=1 ,rotate=-290.48,xslant=-0.01]  {$\mathcal{C}_{D|ABC}^{2}$};
\draw (222.59,125.09) node [anchor=north west][inner sep=0.75pt]  [color={rgb, 255:red, 0; green, 0; blue, 0 }  ,opacity=1 ,rotate=-283.44,xslant=0.02]  {$\mathcal{C}_{A|BCD}^{2}$};
\draw (387.55,52.67) node [anchor=north west][inner sep=0.75pt]  [color={rgb, 255:red, 0; green, 0; blue, 0 }  ,opacity=1 ,rotate=-69.37,xslant=0]  {$\mathcal{C}_{B|ACD}^{2}$};
\draw (282.48,150.79) node [anchor=north west][inner sep=0.75pt]  [color={rgb, 255:red, 0; green, 0; blue, 0 }  ,opacity=1 ,rotate=-358.48,xslant=0.01]  {$\mathcal{C}_{C|ABD}^{2}$};
\draw (272.89,33.5) node [anchor=north west][inner sep=0.75pt]  [color={rgb, 255:red, 0; green, 0; blue, 0 }  ,opacity=1 ,rotate=-334.59,xslant=-0.01]  {$\mathcal{C}_{D|ABC}^{2}$};
\draw (470.74,14.86) node [anchor=north west][inner sep=0.75pt]  [color={rgb, 255:red, 0; green, 0; blue, 0 }  ,opacity=1 ,rotate=-7.43,xslant=0.02]  {$\mathcal{C}_{A|BCD}^{2}$};
\draw (572.09,55.78) node [anchor=north west][inner sep=0.75pt]  [color={rgb, 255:red, 0; green, 0; blue, 0 }  ,opacity=1 ,rotate=-54.52,xslant=0]  {$\mathcal{C}_{B|ACD}^{2}$};
\draw (498.62,144.6) node [anchor=north west][inner sep=0.75pt]  [color={rgb, 255:red, 0; green, 0; blue, 0 }  ,opacity=1 ,rotate=-358.96,xslant=0.01]  {$\mathcal{C}_{C|ABD}^{2}$};
\draw (416.85,109.31) node [anchor=north west][inner sep=0.75pt]  [color={rgb, 255:red, 0; green, 0; blue, 0 }  ,opacity=1 ,rotate=-277.92,xslant=-0.01]  {$\mathcal{C}_{D|ABC}^{2}$};
\draw (348.49,275.15) node [anchor=north west][inner sep=0.75pt]  [rotate=-47.28,xslant=-0.07]  {$C_{AB|CD\ }^{2} =\mathcal{C}_{B|ACD}^{2}$};
\draw (292.09,416.32) node [anchor=north west][inner sep=0.75pt]  [rotate=-359.28,xslant=0.04]  {$\mathcal{C}_{AC|BD}^{2} =\mathcal{C}_{C|ABD}^{2}$};
\draw (259.06,402.12) node [anchor=north west][inner sep=0.75pt]  [rotate=-283.26,xslant=-0.07]  {$\mathcal{C}_{AD|BC}^{2} =\mathcal{C}_{D|ABC}^{2}$};
\draw (457.17,90.01) node [anchor=north west][inner sep=0.75pt]  [rotate=-314.73]  {$\mathcal{C}_{AD|BC}^{2}$};
\draw (305.48,65.3) node [anchor=north west][inner sep=0.75pt]  [rotate=-31.77]  {$\mathcal{C}_{AC|BD}^{2}$};
\draw (127.68,60.77) node [anchor=north west][inner sep=0.75pt]  [rotate=-41.79]  {$\mathcal{C}_{AB|CD}^{2}$};
\draw (82.44,268.72) node [anchor=north west][inner sep=0.75pt]  [color={rgb, 255:red, 0; green, 0; blue, 0 }  ,opacity=1 ,rotate=-359.33,xslant=0.02]  {$\mathcal{C}_{A|BCD}^{2} =0$};
\draw (227.44,217.72) node [anchor=north west][inner sep=0.75pt]  [color={rgb, 255:red, 0; green, 0; blue, 0 }  ,opacity=1 ,rotate=-359.33,xslant=0.02]  {$\mathcal{C}_{A|BCD}^{2} =0$};
\draw (433.44,227.72) node [anchor=north west][inner sep=0.75pt]  [color={rgb, 255:red, 0; green, 0; blue, 0 }  ,opacity=1 ,rotate=-359.33,xslant=0.02]  {$\mathcal{C}_{A|BCD}^{2} =0$};

\end{tikzpicture}

%% file: images/diagonal_zero.tex
\begin{tikzpicture}[x=0.75pt,y=0.75pt,yscale=-1,xscale=1]

\draw [fill={rgb, 255:red, 255; green, 255; blue, 170 }  ,fill opacity=0.83 ]   (619,128.38) -- (552.27,41.93) -- (466.78,128.93) ;
\draw [fill={rgb, 255:red, 255; green, 255; blue, 170 }  ,fill opacity=0.83 ]   (389.91,127.37) -- (351.88,23.95) -- (272.61,59.12) ;
\draw [fill={rgb, 255:red, 255; green, 255; blue, 170 }  ,fill opacity=0.83 ]   (22.91,125.04) -- (53.86,33.08) -- (159.3,128.76) ;
\draw [color={rgb, 255:red, 74; green, 144; blue, 226 }  ,draw opacity=1 ][fill={rgb, 255:red, 255; green, 255; blue, 170 }  ,fill opacity=0.83 ]   (389.91,127.37) -- (272.61,59.12) -- (255.93,129.25) ;
\draw [color={rgb, 255:red, 0; green, 0; blue, 0 }  ,draw opacity=1 ][fill={rgb, 255:red, 255; green, 255; blue, 170 }  ,fill opacity=0.83 ][line width=1.5]    (272.61,59.12) -- (351.88,23.95) ;
\draw [color={rgb, 255:red, 0; green, 0; blue, 0 }  ,draw opacity=1 ][fill={rgb, 255:red, 255; green, 255; blue, 170 }  ,fill opacity=0.83 ][line width=1.5]    (389.91,127.37) -- (255.93,129.25) ;
\draw [color={rgb, 255:red, 0; green, 0; blue, 0 }  ,draw opacity=1 ][fill={rgb, 255:red, 255; green, 255; blue, 170 }  ,fill opacity=0.83 ][line width=1.5]    (351.88,23.95) -- (389.91,127.37) ;
\draw [color={rgb, 255:red, 0; green, 0; blue, 0 }  ,draw opacity=1 ][fill={rgb, 255:red, 255; green, 255; blue, 170 }  ,fill opacity=0.83 ][line width=1.5]    (255.93,129.25) -- (272.61,59.12) ;
\draw [color={rgb, 255:red, 74; green, 144; blue, 226 }  ,draw opacity=1 ][fill={rgb, 255:red, 255; green, 255; blue, 170 }  ,fill opacity=0.83 ]   (53.86,33.08) -- (159.3,128.76) -- (156.54,52.07) ;
\draw [color={rgb, 255:red, 0; green, 0; blue, 0 }  ,draw opacity=1 ][fill={rgb, 255:red, 255; green, 255; blue, 170 }  ,fill opacity=0.83 ][line width=1.5]    (53.86,33.08) -- (22.91,125.04) ;
\draw [color={rgb, 255:red, 0; green, 0; blue, 0 }  ,draw opacity=1 ][fill={rgb, 255:red, 255; green, 255; blue, 170 }  ,fill opacity=0.83 ][line width=1.5]    (159.3,128.76) -- (22.91,125.04) ;
\draw [color={rgb, 255:red, 0; green, 0; blue, 0 }  ,draw opacity=1 ][fill={rgb, 255:red, 255; green, 255; blue, 170 }  ,fill opacity=0.83 ][line width=1.5]    (156.54,52.07) -- (53.86,33.08) ;
\draw [color={rgb, 255:red, 0; green, 0; blue, 0 }  ,draw opacity=1 ][fill={rgb, 255:red, 255; green, 255; blue, 170 }  ,fill opacity=0.83 ][line width=1.5]    (156.54,52.07) -- (159.3,128.76) ;
\draw [color={rgb, 255:red, 74; green, 144; blue, 226 }  ,draw opacity=1 ][fill={rgb, 255:red, 255; green, 255; blue, 170 }  ,fill opacity=0.83 ]   (466.78,128.93) -- (552.27,41.93) -- (478.39,31.1) ;
\draw [color={rgb, 255:red, 0; green, 0; blue, 0 }  ,draw opacity=1 ][fill={rgb, 255:red, 255; green, 255; blue, 170 }  ,fill opacity=0.83 ][line width=1.5]    (478.39,31.1) -- (466.78,128.93) ;
\draw [color={rgb, 255:red, 0; green, 0; blue, 0 }  ,draw opacity=1 ][fill={rgb, 255:red, 255; green, 255; blue, 170 }  ,fill opacity=0.83 ][line width=1.5]    (466.78,128.93) -- (619,128.38) ;
\draw [color={rgb, 255:red, 0; green, 0; blue, 0 }  ,draw opacity=1 ][fill={rgb, 255:red, 255; green, 255; blue, 170 }  ,fill opacity=0.83 ][line width=1.5]    (552.27,41.93) -- (619,128.38) ;
\draw [color={rgb, 255:red, 0; green, 0; blue, 0 }  ,draw opacity=1 ][fill={rgb, 255:red, 255; green, 255; blue, 170 }  ,fill opacity=0.83 ][line width=1.5]    (478.39,31.1) -- (552.27,41.93) ;
\draw    (316,151) -- (316.98,232) ;
\draw [shift={(317,234)}, rotate = 269.31] [color={rgb, 255:red, 0; green, 0; blue, 0 }  ][line width=0.75]    (10.93,-3.29) .. controls (6.95,-1.4) and (3.31,-0.3) .. (0,0) .. controls (3.31,0.3) and (6.95,1.4) .. (10.93,3.29)   ;
\draw    (537,154) -- (537.98,235) ;
\draw [shift={(538,237)}, rotate = 269.31] [color={rgb, 255:red, 0; green, 0; blue, 0 }  ][line width=0.75]    (10.93,-3.29) .. controls (6.95,-1.4) and (3.31,-0.3) .. (0,0) .. controls (3.31,0.3) and (6.95,1.4) .. (10.93,3.29)   ;
\draw    (90,151) -- (90.98,232) ;
\draw [shift={(91,234)}, rotate = 269.31] [color={rgb, 255:red, 0; green, 0; blue, 0 }  ][line width=0.75]    (10.93,-3.29) .. controls (6.95,-1.4) and (3.31,-0.3) .. (0,0) .. controls (3.31,0.3) and (6.95,1.4) .. (10.93,3.29)   ;
\draw [color={rgb, 255:red, 0; green, 0; blue, 0 }  ,draw opacity=1 ][fill={rgb, 255:red, 255; green, 255; blue, 170 }  ,fill opacity=0.83 ][line width=1.5]    (83,288) -- (20.28,287.41) ;
\draw [color={rgb, 255:red, 0; green, 0; blue, 0 }  ,draw opacity=1 ][fill={rgb, 255:red, 255; green, 255; blue, 170 }  ,fill opacity=0.83 ][line width=1.5]    (83,288) -- (20.28,287.41) ;
\draw [color={rgb, 255:red, 0; green, 0; blue, 0 }  ,draw opacity=1 ][fill={rgb, 255:red, 255; green, 255; blue, 170 }  ,fill opacity=0.83 ][line width=1.5]    (174,288) -- (83,288) ;
\draw [color={rgb, 255:red, 0; green, 0; blue, 0 }  ,draw opacity=1 ][fill={rgb, 255:red, 255; green, 255; blue, 170 }  ,fill opacity=0.83 ][line width=1.5]    (174,288) -- (83,288) ;
\draw  [fill={rgb, 255:red, 0; green, 0; blue, 0 }  ,fill opacity=1 ] (83,288) .. controls (83,286.62) and (84.12,285.5) .. (85.5,285.5) .. controls (86.88,285.5) and (88,286.62) .. (88,288) .. controls (88,289.38) and (86.88,290.5) .. (85.5,290.5) .. controls (84.12,290.5) and (83,289.38) .. (83,288) -- cycle ;
\draw  [fill={rgb, 255:red, 0; green, 0; blue, 0 }  ,fill opacity=1 ] (20.28,287.41) .. controls (20.28,286.03) and (21.4,284.91) .. (22.78,284.91) .. controls (24.16,284.91) and (25.28,286.03) .. (25.28,287.41) .. controls (25.28,288.79) and (24.16,289.91) .. (22.78,289.91) .. controls (21.4,289.91) and (20.28,288.79) .. (20.28,287.41) -- cycle ;
\draw  [fill={rgb, 255:red, 0; green, 0; blue, 0 }  ,fill opacity=1 ] (174,288) .. controls (174,286.62) and (175.12,285.5) .. (176.5,285.5) .. controls (177.88,285.5) and (179,286.62) .. (179,288) .. controls (179,289.38) and (177.88,290.5) .. (176.5,290.5) .. controls (175.12,290.5) and (174,289.38) .. (174,288) -- cycle ;
\draw [fill={rgb, 255:red, 255; green, 255; blue, 170 }  ,fill opacity=0.83 ]   (347,335) -- (364,266) -- (287.61,265.12) ;
\draw [color={rgb, 255:red, 74; green, 144; blue, 226 }  ,draw opacity=1 ][fill={rgb, 255:red, 255; green, 255; blue, 170 }  ,fill opacity=0.83 ]   (347,335) -- (287.61,265.12) -- (270.93,335.25) ;
\draw [color={rgb, 255:red, 0; green, 0; blue, 0 }  ,draw opacity=1 ][fill={rgb, 255:red, 255; green, 255; blue, 170 }  ,fill opacity=0.83 ][line width=1.5]    (287.61,265.12) -- (364,266) ;
\draw [color={rgb, 255:red, 0; green, 0; blue, 0 }  ,draw opacity=1 ][fill={rgb, 255:red, 255; green, 255; blue, 170 }  ,fill opacity=0.83 ][line width=1.5]    (347,335) -- (270.93,335.25) ;
\draw [color={rgb, 255:red, 0; green, 0; blue, 0 }  ,draw opacity=1 ][fill={rgb, 255:red, 255; green, 255; blue, 170 }  ,fill opacity=0.83 ][line width=1.5]    (364,266) -- (347,335) ;
\draw [color={rgb, 255:red, 0; green, 0; blue, 0 }  ,draw opacity=1 ][fill={rgb, 255:red, 255; green, 255; blue, 170 }  ,fill opacity=0.83 ][line width=1.5]    (270.93,335.25) -- (287.61,265.12) ;
\draw [fill={rgb, 255:red, 255; green, 255; blue, 170 }  ,fill opacity=0.83 ]   (585,267) -- (491.93,336.25) -- (508.61,266.12) ;
\draw [color={rgb, 255:red, 74; green, 144; blue, 226 }  ,draw opacity=1 ][fill={rgb, 255:red, 255; green, 255; blue, 170 }  ,fill opacity=0.83 ]   (568,336) -- (585,267) -- (491.93,336.25) ;
\draw [color={rgb, 255:red, 0; green, 0; blue, 0 }  ,draw opacity=1 ][fill={rgb, 255:red, 255; green, 255; blue, 170 }  ,fill opacity=0.83 ][line width=1.5]    (508.61,266.12) -- (585,267) ;
\draw [color={rgb, 255:red, 0; green, 0; blue, 0 }  ,draw opacity=1 ][fill={rgb, 255:red, 255; green, 255; blue, 170 }  ,fill opacity=0.83 ][line width=1.5]    (568,336) -- (491.93,336.25) ;
\draw [color={rgb, 255:red, 0; green, 0; blue, 0 }  ,draw opacity=1 ][fill={rgb, 255:red, 255; green, 255; blue, 170 }  ,fill opacity=0.83 ][line width=1.5]    (585,267) -- (568,336) ;
\draw [color={rgb, 255:red, 0; green, 0; blue, 0 }  ,draw opacity=1 ][fill={rgb, 255:red, 255; green, 255; blue, 170 }  ,fill opacity=0.83 ][line width=1.5]    (491.93,336.25) -- (508.61,266.12) ;

\draw (179.58,63.65) node [anchor=north west][inner sep=0.75pt]  [font=\footnotesize,color={rgb, 255:red, 0; green, 0; blue, 0 }  ,opacity=1 ,rotate=-88.22,xslant=0.02]  {$C_{A|BCD}^{2}$};
\draw (89.08,18.46) node [anchor=north west][inner sep=0.75pt]  [font=\footnotesize,color={rgb, 255:red, 0; green, 0; blue, 0 }  ,opacity=1 ,rotate=-8.74,xslant=0]  {$C_{B|ACD}^{2}$};
\draw (66.15,128.49) node [anchor=north west][inner sep=0.75pt]  [font=\footnotesize,color={rgb, 255:red, 0; green, 0; blue, 0 }  ,opacity=1 ,rotate=-0.18,xslant=0.01]  {$C_{C|ABD}^{2}$};
\draw (8.79,95.54) node [anchor=north west][inner sep=0.75pt]  [font=\footnotesize,color={rgb, 255:red, 0; green, 0; blue, 0 }  ,opacity=1 ,rotate=-290.48,xslant=-0.01]  {$C_{D|ABC}^{2}$};
\draw (235.63,112.22) node [anchor=north west][inner sep=0.75pt]  [font=\footnotesize,color={rgb, 255:red, 0; green, 0; blue, 0 }  ,opacity=1 ,rotate=-283.44]  {$C_{A|BCD}^{2}$};
\draw (383.21,42.93) node [anchor=north west][inner sep=0.75pt]  [font=\footnotesize,color={rgb, 255:red, 0; green, 0; blue, 0 }  ,opacity=1 ,rotate=-69.37]  {$C_{B|ACD}^{2}$};
\draw (288.27,131.93) node [anchor=north west][inner sep=0.75pt]  [font=\footnotesize,color={rgb, 255:red, 0; green, 0; blue, 0 }  ,opacity=1 ,rotate=-358.48]  {$C_{C|ABD}^{2}$};
\draw (278.73,32.1) node [anchor=north west][inner sep=0.75pt]  [font=\footnotesize,color={rgb, 255:red, 0; green, 0; blue, 0 }  ,opacity=1 ,rotate=-334.59]  {$C_{D|ABC}^{2}$};
\draw (493.37,12.73) node [anchor=north west][inner sep=0.75pt]  [font=\footnotesize,color={rgb, 255:red, 0; green, 0; blue, 0 }  ,opacity=1 ,rotate=-7.43,xslant=0.02]  {$C_{A|BCD}^{2}$};
\draw (582.94,46.6) node [anchor=north west][inner sep=0.75pt]  [font=\footnotesize,color={rgb, 255:red, 0; green, 0; blue, 0 }  ,opacity=1 ,rotate=-54.52,xslant=0]  {$C_{B|ACD}^{2}$};
\draw (517.06,131.13) node [anchor=north west][inner sep=0.75pt]  [font=\footnotesize,color={rgb, 255:red, 0; green, 0; blue, 0 }  ,opacity=1 ,rotate=-358.96,xslant=0.01]  {$C_{C|ABD}^{2}$};
\draw (448.51,99.9) node [anchor=north west][inner sep=0.75pt]  [font=\footnotesize,color={rgb, 255:red, 0; green, 0; blue, 0 }  ,opacity=1 ,rotate=-277.92,xslant=-0.01]  {$C_{D|ABC}^{2}$};
\draw (481.49,81.37) node [anchor=north west][inner sep=0.75pt]  [font=\footnotesize,rotate=-314.73]  {$C_{AD|BC}^{2}$};
\draw (308.77,54.44) node [anchor=north west][inner sep=0.75pt]  [font=\footnotesize,rotate=-31.77]  {$C_{AC|BD}^{2}$};
\draw (103.04,49.21) node [anchor=north west][inner sep=0.75pt]  [font=\footnotesize,rotate=-41.79]  {$C_{AB|CD}^{2}$};
\draw (23.44,178.72) node [anchor=north west][inner sep=0.75pt]  [font=\footnotesize,color={rgb, 255:red, 0; green, 0; blue, 0 }  ,opacity=1 ,rotate=-359.33,xslant=0.02]  {$C_{AB|CD}^{2} =0$};
\draw (248.44,183.72) node [anchor=north west][inner sep=0.75pt]  [font=\footnotesize,color={rgb, 255:red, 0; green, 0; blue, 0 }  ,opacity=1 ,rotate=-359.33,xslant=0.02]  {$C_{AB|CD}^{2} =0$};
\draw (470.44,182.72) node [anchor=north west][inner sep=0.75pt]  [font=\footnotesize,color={rgb, 255:red, 0; green, 0; blue, 0 }  ,opacity=1 ,rotate=-359.33,xslant=0.02]  {$C_{AB|CD}^{2} =0$};
\draw (32.81,290.26) node [anchor=north west][inner sep=0.75pt]  [font=\footnotesize,color={rgb, 255:red, 0; green, 0; blue, 0 }  ,opacity=1 ,rotate=-358.58,xslant=0.02]  {$C_{A|BCD}^{2}$};
\draw (32.28,267.75) node [anchor=north west][inner sep=0.75pt]  [font=\footnotesize,color={rgb, 255:red, 0; green, 0; blue, 0 }  ,opacity=1 ,rotate=-358.36,xslant=0]  {$C_{B|ACD}^{2}$};
\draw (108.61,290.87) node [anchor=north west][inner sep=0.75pt]  [font=\footnotesize,color={rgb, 255:red, 0; green, 0; blue, 0 }  ,opacity=1 ,rotate=-359.67,xslant=0.01]  {$C_{C|ABD}^{2}$};
\draw (110.39,267.53) node [anchor=north west][inner sep=0.75pt]  [font=\footnotesize,color={rgb, 255:red, 0; green, 0; blue, 0 }  ,opacity=1 ,rotate=-0.95,xslant=-0.01]  {$C_{D|ABC}^{2}$};
\draw (250.63,318.22) node [anchor=north west][inner sep=0.75pt]  [font=\footnotesize,color={rgb, 255:red, 0; green, 0; blue, 0 }  ,opacity=1 ,rotate=-283.44]  {$C_{A|BCD}^{2}$};
\draw (380.99,281.03) node [anchor=north west][inner sep=0.75pt]  [font=\footnotesize,color={rgb, 255:red, 0; green, 0; blue, 0 }  ,opacity=1 ,rotate=-104.05]  {$C_{B|ACD}^{2}$};
\draw (287.27,340.93) node [anchor=north west][inner sep=0.75pt]  [font=\footnotesize,color={rgb, 255:red, 0; green, 0; blue, 0 }  ,opacity=1 ,rotate=-358.48]  {$C_{C|ABD}^{2}$};
\draw (303.03,245.67) node [anchor=north west][inner sep=0.75pt]  [font=\footnotesize,color={rgb, 255:red, 0; green, 0; blue, 0 }  ,opacity=1 ,rotate=-359.02]  {$C_{D|ABC}^{2}$};
\draw (319.84,272.84) node [anchor=north west][inner sep=0.75pt]  [font=\footnotesize,rotate=-52.07]  {$C_{AC|BD}^{2}$};
\draw (521.31,247.8) node [anchor=north west][inner sep=0.75pt]  [font=\footnotesize,color={rgb, 255:red, 0; green, 0; blue, 0 }  ,opacity=1 ,rotate=-359.38]  {$C_{A|BCD}^{2}$};
\draw (601.99,282.03) node [anchor=north west][inner sep=0.75pt]  [font=\footnotesize,color={rgb, 255:red, 0; green, 0; blue, 0 }  ,opacity=1 ,rotate=-104.05]  {$C_{B|ACD}^{2}$};
\draw (508.27,341.93) node [anchor=north west][inner sep=0.75pt]  [font=\footnotesize,color={rgb, 255:red, 0; green, 0; blue, 0 }  ,opacity=1 ,rotate=-358.48]  {$C_{C|ABD}^{2}$};
\draw (475.22,317.82) node [anchor=north west][inner sep=0.75pt]  [font=\footnotesize,color={rgb, 255:red, 0; green, 0; blue, 0 }  ,opacity=1 ,rotate=-284.03]  {$C_{D|ABC}^{2}$};
\draw (514.61,295.19) node [anchor=north west][inner sep=0.75pt]  [font=\footnotesize,rotate=-323.97]  {$C_{AD|BC}^{2}$};

\end{tikzpicture}